\documentclass{llncs}

\usepackage{etex}
\usepackage{times}  
\usepackage{url}  
\usepackage{graphicx}  
\usepackage{xspace}
\usepackage{amssymb,amsmath,amsfonts}
\usepackage{autobreak}
\usepackage{mathtools}

\newcommand{\pto}{\xrightharpoondown{}}
\newcommand{\set}[1]{\left\{ #1 \right\}}

\DeclareMathOperator{\FRuns}{Paths}
\newcommand{\Runs}{\Omega}
\newcommand{\ETotal}{\mathsf{ETotal}}
\newcommand{\Total}{\mathsf{Total}}
\newcommand{\EDisct}{\mathsf{EDisct}}
\newcommand{\Disct}{\mathsf{Disct}}

\DeclareMathOperator{\PSat}{\mathsf{PSyn}}
\DeclareMathOperator{\PSemSat}{\mathsf{PSem}}

\newcommand{\f}{{\mu}} 
\newcommand{\Aa}{\mathcal{A}}

\newcommand{\Mm}{\mathcal{M}}

\newcommand{\Ff}{\mathcal{F}}

\newcommand{\eE}{\mathbb E}

\newcommand{\Real}{\mathbb R}

\newcommand{\DIST}{{\cal D}}

\DeclareMathOperator{\supp}{supp}
\DeclareMathOperator{\last}{last}
\DeclareMathOperator{\infi}{inf}

\title{Reward Shaping for Reinforcement Learning with Omega-Regular Objectives}
\author{Ernst Moritz Hahn\inst{1,2} \and Mateo Perez\inst{3}
  \and Sven Schewe\inst{4} \and \\  Fabio Somenzi\inst{3} \and
  Ashutosh Trivedi\inst{3} \and Dominik Wojtczak\inst{4}}

\institute{
  School of EEECS, Queen’s University Belfast, UK
  \and
  State Key Laboratory of Computer Science, Institute of Software, CAS, PRC
  \and
  University of Colorado Boulder, USA
  \and
  University of Liverpool, UK
}
\authorrunning{E. M. Hahn, M. Perez, S. Schewe, F. Somenzi, A. Trivedi, and D. Wojtczak}

\begin{document}
\maketitle

\begin{abstract}
  Recently, successful approaches have been made to exploit good-for-MDPs automata---B\"uchi automata with a restricted form of nondeterminism---for model free reinforcement learning, a class of automata that subsumes good for games automata and the most widespread class of limit deterministic automata~\cite{Hahn20}.
  The foundation of using these B\"uchi automata is that the B\"uchi condition can, for good-for-MDP automata, be translated to reachability~\cite{Hahn19}.
  
  The drawback of this translation is that the rewards are, on average, reaped very late, which requires long episodes during the learning process. We devise a new reward shaping approach that overcomes this issue.
  We show that the resulting a model is equivalent to a discounted payoff objective with a biased discount that simplifies and improves 
  on \cite{DBLP:journals/corr/abs-1909-07299}.

\end{abstract}

\section{Preliminaries}

A \emph{nondeterministic B\"uchi automaton} is a tuple
${\cal A} = \langle \Sigma,Q,q_0,\Delta,\Gamma \rangle$, where
$\Sigma$ is a finite \emph{alphabet}, $Q$ is a finite set of
\emph{states}, $q_0 \in Q$ is the \emph{initial state},
$\Delta \subseteq Q \times \Sigma \times Q$ are transitions, and
$\Gamma \subseteq Q \times \Sigma \times Q$ is the transition-based
\emph{acceptance condition}.

A \emph{run} $r$ of ${\cal A}$ on $w \in \Sigma^\omega$ is an
$\omega$-word $r_0, w_0, r_1, w_1, \ldots$ in
$(Q \times \Sigma)^\omega$ such that $r_0 = q_0$ and, for $i > 0$, it
is $(r_{i-1},w_{i-1},r_i)\in \Delta$.  We write $\infi(r)$ for the set
of transitions that appear infinitely often in the run $r$.  A run $r$
of ${\cal A}$ is \emph{accepting} if
$\infi(r) \cap \Gamma \neq \emptyset$.

The \emph{language}, $L_{\mathcal{A}}$, of ${\cal A}$ (or,
\emph{recognized} by ${\cal A}$) is the subset of words in
$\Sigma^\omega$ that have accepting runs in ${\cal A}$.  A language is
$\omega$-\emph{regular} if it is accepted by a B\"uchi automaton.  An
automaton ${\cal A} = \langle\Sigma,Q,Q_0,\Delta,\Gamma\rangle$ is
\emph{deterministic} if $(q,\sigma,q'),(q,\sigma,q'') \in \Delta$
implies $q'=q''$.  ${\cal A}$ is \emph{complete} if, for all
$\sigma \in \Sigma$ and all $q \in Q$, there is a transition
$(q,\sigma,q')\in \Delta$.  A word in $\Sigma^\omega$ has exactly one
run in a deterministic, complete automaton.

A \emph{Markov decision process (MDP)} $\Mm$ is a tuple
$(S, A, T, \Sigma, L)$ where $S$ is a finite set of states, $A$ is a
finite set of \emph{actions}, $T: S \times A \to \DIST(S)$, where
$\DIST(S)$ is the set of probability distributions over $S$, is the
\emph{probabilistic transition function}, $\Sigma$ is an alphabet, and
$L: S \times A \times S \to \Sigma$ is the \emph{labelling function} of
the set of transitions.  For a state $s \in S$, $A(s)$ denotes the set
of actions available in $s$.  For states $s, s' \in S$ and
$a \in A(s)$, we have that $T(s, a)(s')$ equals $\Pr{}(s' | s, a)$.

A \emph{run} of $\Mm$ is an $\omega$-word
$s_0, a_1, \ldots \in S \times (A \times S)^\omega$ such that
$\Pr{}(s_{i+1} | s_{i}, a_{i+1}) > 0$ for all $i \geq 0$.  A finite
run is a finite such sequence.  For a \emph{run}
$r = s_0, a_1, s_1, \ldots$ we define the corresponding labelled
run as
$L(r) = L(s_0,a_1,s_1), L(s_1,a_2,s_2), \ldots \in
\Sigma^\omega$.  We write $\Runs(\Mm)$ ($\FRuns(\Mm)$) for the set of
runs (finite runs) of $\Mm$ and $\Runs_s(\Mm)$ ($\FRuns_s(\Mm)$) for
the set of runs (finite runs) of $\Mm$ starting from state $s$.  When
the MDP is clear from the context we drop the argument $\Mm$.

A strategy in $\Mm$ is a function $\f : \FRuns \to \DIST(A)$ that for all finite runs $r$ we have
$\supp(\f(r)) \subseteq A(\last(r))$, where $\supp(d)$ is the support
of $d$ and $\last(r)$ is the last state of $r$.  Let $\Runs^\f_s(\Mm)$
denote the subset of runs $\Runs_s(\Mm)$ that correspond to strategy
$\f$ and initial state $s$.  Let $\Sigma_\Mm$ be the set of all
strategies.  We say that a strategy $\f$ is \emph{pure} if $\f(r)$ is
a point distribution for all runs $r \in \FRuns$ and we say that
$\f$ is \emph{positional} if $\last(r) = \last(r')$ implies
$\f(r) = \f(r')$ for all runs $r, r' \in \FRuns$.

The behaviour of an MDP $\Mm$ under a strategy $\f$ with starting state
$s$ is defined on a probability space
$(\Runs^\f_s, \Ff^\f_s, \Pr^\f_s)$ over the set of infinite runs of
$\f$ from $s$.  Given a random variable over the set of infinite runs
$f :\Runs \to \Real$, we write $\eE^\f_s \set{f}$ for the expectation
of $f$ over the runs of $\Mm$ from state $s$ that follow strategy
$\f$.

Given an MDP $\Mm$ and an automaton
$\Aa = \langle \Sigma, Q, q_0, \Delta, \Gamma \rangle$, we want to
compute an optimal strategy satisfying the objective that the run of
$\Mm$ is in the language of $\Aa$.  We define the semantic
satisfaction probability for $\Aa$ and a strategy $\f$ from state $s$
as:
\begin{align*}
\PSemSat^\Mm_{\Aa}(s, \f) &= \Pr{}_s^\f \{ r {\in} \Runs^\f_s :
  L(r) {\in} L_{\mathcal{A}} \} \text{ and} &
  \PSemSat^\Mm_{\Aa}(s) &= \sup_{\f}\big(\PSemSat^\Mm_{\Aa}(s, \f)\big)
  \,.
\end{align*}
When using automata for the analysis of MDPs, we need a syntactic
variant of the acceptance condition.  Given an MDP
$\Mm = (S, A, T, \Sigma, L)$ with initial state $s_0 \in S$
and an automaton
$\mathcal{A} = \langle \Sigma, Q, q_0, \Delta, \Gamma \rangle$, the
\emph{product}
$\Mm \times \mathcal{A} = (S \times Q, (s_0,q_0), A \times Q,
T^\times, \Gamma^\times)$ is an MDP augmented with an initial state $(s_0,q_0)$
and accepting transitions $\Gamma^\times$.  The function
$T^\times : (S \times Q) \times (A \times Q) \pto \DIST(S \times Q)$
is defined by
\begin{equation*}
  T^\times((s,q),(a,q'))(({s}',{q}')) =
  \begin{cases}
    T(s,a)({s}') & \text{if } (q,L(s,a,{s}'),{q}') \in \Delta \\
    0 & \text{otherwise.}
  \end{cases}
\end{equation*}
Finally,
$\Gamma^\times \subseteq (S \times Q) \times (A \times Q) \times (S
\times Q)$ is defined by $((s,q),(a,q'),(s',q')) \in \Gamma^\times$
if, and only if, $(q,L(s,a,s'),q') \in \Gamma$ and $T(s,a)(s') > 0$.
A strategy $\f$ on the MDP defines a strategy $\f^\times$ on the
product, and vice versa.  We define the syntactic satisfaction
probabilities as
\begin{align*}
  \PSat^\Mm_{\Aa}((s,q), \f^\times) &= \Pr{}_s^\f \{ r \in
  \Runs^{\f^\times}_{(s,q)}(\Mm\times\Aa) : \inf(r) \cap \Gamma^\times
  \neq \emptyset \} \enspace, ~~~~\text{ and} \\
  \PSat^\Mm_{\Aa}(s) &= \sup_{\f^\times}\big(\PSat^\Mm_{\Aa}((s,q_0),
  \f^\times)\big) \enspace.
\end{align*}
Note that $\PSat^\Mm_{\Aa}(s) = \PSemSat^\Mm_{\Aa}(s)$ holds for a
deterministic $\Aa$.  In general,
$\PSat^\Mm_{\Aa}(s)$ $\leq \PSemSat^\Mm_{\cal A}(s)$ holds, but equality
is not guaranteed because the optimal resolution of nondeterministic
choices may require access to future events.

\begin{definition}[GFM automata \cite{Hahn20}]
  \label{def:gfm}
  An automaton $\Aa$ is \emph{good for MDPs} if, for all
  MDPs $\Mm$, $\PSat^\Mm_{\Aa}(s_0) = \PSemSat^\Mm_{\Aa}(s_0)$ holds,
  where $s_0$ is the initial state of $\Mm$.
\end{definition}
For an automaton to match $\PSemSat^\Mm_{\Aa}(s_0)$,
its nondeterminism is restricted not to rely heavily on the future;
rather, it must be possible to resolve the nondeterminism on-the-fly.

\section{Undiscounted Reward Shaping}
\label{shape}
We build on the reduction from \cite{Hahn19,Hahn20} that reduces maximising the chance to realise an $\omega$-regular objective given by a good-for-MDPs B\"uchi automaton $\Aa$ for an MDP $\Mm$ to maximising the chance to meet the reachability objective in the augmented MDP ${\Mm}^\zeta$ (for $\zeta \in ]0,1[$) obtained from ${\Mm} \times {\Aa}$ by
\begin{itemize}
 \item adding a new target state $t$ (either as a sink with a self-loop or as a point where the computation stops; we choose here the latter view) and
 \item by making the target $t$ a destination of each accepting transition $\tau$ of ${\Mm} \times {\Aa}$ with probability $1-\zeta$ and
 \newline multiplying the original probabilities of all other destinations of an accepting transition $\tau$ by $\zeta$.
\end{itemize}

Let
\begin{align*}
  \PSat^{\Mm^\zeta}_{t}((s,q), \f) &= \Pr{}_s^\f \{ r \in
  \Runs^{\f}_{(s,q)}(\Mm^\zeta) : r \mbox{ reaches } t \} \enspace, ~~~~\text{ and} \\
  \PSat^{\Mm^\zeta}_{t}(s) &= \sup_{\f}\big(\PSat^{\Mm^\zeta}_{t}((s,q_0),
  \f)\big) \enspace.
\end{align*}

\begin{theorem}[\cite{Hahn19,Hahn20}]
\label{theo:reach}
The following holds:
\begin{enumerate}
 \item 
 ${\Mm}^\zeta$ (for $\zeta \in ]0,1[$) and ${\Mm} \times {\Aa}$ have the same set of strategies.
 \item
 For a strategy $\f$, the chance of reaching the target $t$ in ${\Mm}^\zeta_\f$ is $1$ if, and only if, the chance of satisfying the B\"uchi objective in $({\Mm} \times {\Aa})_\f$ is $1$:
 
 $\PSat^{\Mm^\zeta}_{t}((s_0,q_0), \f) = 1 \; \Leftrightarrow \; \PSat^{\Mm}_{\Aa}(s_0,q_0), \f) = 1$
 \item
 There is a $\zeta_0 \in ]0,1[$ such that, for all $\zeta \in [\zeta_0,1[$, an optimal reachability strategy $\f$ for  ${\Mm}^\zeta$ is an optimal strategy for satisfying the B\"uchi objective in ${\Mm} \times {\Aa}$:
 
 $\PSat^{\Mm^\zeta}_{t}((s_0,q_0),\f)  = \PSat^{\Mm^\zeta}_{t}(s_0) \; \Rightarrow \; \PSat^{\Mm}_{\Aa}(s_0,q_0), \f) = \PSat^{\Mm}_{\Aa}(s_0))$.
\end{enumerate}
\end{theorem}

This allows for analysing the much simpler reachability objective in ${\Mm}^\zeta_\f$ instead of the B\"uchi objective in ${\Mm} \times {\Aa}$, and is open to implementation in model free reinforcement learning.

However, it has the drawback that rewards occur late when $\zeta$ is close to $1$.
We amend that by the following observation:

We build, for a good-for-MDPs B\"uchi automaton $\Aa$ and an MDP $\Mm$, the augmented MDP $\overline{\Mm}^\zeta$ (for $\zeta \in ]0,1[$) obtained from ${\Mm} \times {\Aa}$ in the same way as ${\Mm}^\zeta$, i.e.\ by
\begin{itemize}
 \item adding a new sink state $t$ (as a sink where the computation stops) and
 \item by making the sink $t$ a destination of each accepting transition $\tau$ of ${\Mm} \times {\Aa}$ with probability $1-\zeta$ and
 \newline multiplying the original probabilities of all other destinations of an accepting transition $\tau$ by $\zeta$.
\end{itemize}
Different to ${\Mm}^\zeta$, $\overline{\Mm}^\zeta$ has an undiscounted reward objective, where taking an accepting (in ${\Mm} \times {\Aa}$) transition $\tau$ provides a reward of $1$, regardless of whether it leads to the sink $t$ or stays in the state-space of ${\Mm} \times {\Aa}$.

Let, for a run $r$ of ${\Mm}^\zeta$ that contains $n \in {\mathbb N}_0 \cup \{\infty\}$ accepting transitions, the total reward be $\Total(r) = n$, and let
\begin{align*}
\ETotal^{\overline{\Mm}^\zeta}((s,q), \f) &= 
{\mathbb E}_s^\f \{ \Total(r) : r \in \Runs^{\f}_{(s,q)}(\overline{\Mm}^\zeta) \}   \enspace, ~~~~\text{ and} \\
  \ETotal^{\overline{\Mm}^\zeta}(s) &= \sup_{\f}\big(\ETotal^{\overline{\Mm}^\zeta}((s,q_0),\f)\big) \enspace.
\end{align*}

Note that the set of runs with $\Total(r) = \infty$ has probability $0$ in $\Runs^{\f}_{(s,q)}(\overline{\Mm}^\zeta)$:
they are the runs that infinitely often do not move to $t$ on an accepting transition, where the chance that this happens at least $n$ times is $(1-\zeta)^n$ for all $n \in {\mathbb N}_0$.

\begin{theorem}
\label{theo:undiscounted}
The following holds:
\begin{enumerate}
 \item 
 $\overline{\Mm}^\zeta$ (for $\zeta \in ]0,1[$), ${\Mm}^\zeta$ (for $\zeta \in ]0,1[$), and ${\Mm} \times {\Aa}$ have the same set of strategies.
 \item
 For a strategy $\f$, the expected reward for $\overline{\Mm}^\zeta_\f$ is $r$ if, and only if, the chance of reaching the target $t$ in ${\Mm}^\zeta_\f$ is $\frac{r}{1-\zeta}$:
 
 $\PSat^{\Mm^\zeta}_{t}((s_0,q_0), \f) = (1-\zeta) \ETotal^{\overline{\Mm}^\zeta}((s_0,q_0),\f)$.
 \item The expected reward for $\overline{\Mm}^\zeta_\f$ is in $[0,\frac{1}{1-\zeta}]$.
 \item The chance of satisfying the B\"uchi objective in $({\Mm} \times {\Aa})_\f$ is $1$ if, and only if, the expected reward for $\overline{\Mm}^\zeta_\f$ is $\frac{1}{1-\zeta}$.
 \item
 There is a $\zeta_0 \in ]0,1[$ such that, for all $\zeta \in [\zeta_0,1[$, a strategy $\f$ that maximises the reward for $\overline{\Mm}^\zeta$ is an optimal strategy for satisfying the B\"uchi objective in ${\Mm} \times {\Aa}$.
\end{enumerate}
\end{theorem}
\begin{proof}
(1) Obvious, because all the states and their actions are the same apart from the sink state $t$ for which the strategy can be left undefined.

(2) The sink state $t$ can only be visited once along any run, so the expected number of times a run starting at $(s_0,q_0)$ is going to visit $t$ while using strategy $\f$ is the same as its probability of visiting $t$, i.e., $\PSat^{\Mm^\zeta}_{t}((s_0,q_0), \f)$.
The only way $t$ can be reached is by traversing an accepting transition and this always happens with the same probability $(1-\zeta)$.
So the expected number of visits to $t$ is
the expected number of times an accepting transition is used, i.e., 
$\ETotal^{\overline{\Mm}^\zeta}((s_0,q_0),\f)$, multiplied by $(1-\zeta)$.

(3) follows from (2), because $\PSat^{\Mm^\zeta}_{t}((s_0,q_0), \f)$ cannot be greater than 1.

(4) follows from (2) and Theorem \ref{theo:reach} (2).

(5) follows from (2) and Theorem 
\ref{theo:reach} (3).
\end{proof}

\section{Discounted Reward Shaping}
The expected undiscounted reward for $\overline{\Mm}^\zeta_\f$ can be viewed as a discounted reward for $({\Mm} \times {\Aa})_\f$, by giving a reward $\zeta^i$ to when passing through an accepting transition when $i$ \emph{accepting} transitions have been passed before. We call this reward \emph{$\zeta$-biased}.

Let, for a run $r$ of $\Mm \times \Aa$ that contains $n \in {\mathbb N}_0 \cup \{\infty\}$ accepting transitions, the $\zeta$-biased discounted reward be $\Disct_\zeta(r) = \sum_{i=0}^{n-1} \zeta^i$, and let
\begin{align*}
\EDisct^{\Mm \times \Aa}_\zeta((s,q), \f) &= {\mathbb E}_s^\f \{ r \in
  \Runs^{\f}_{(s,q)}(\Mm \times \Aa) : \Disct_\zeta(r) \}   \enspace, ~~~~\text{ and} \\
  \EDisct^{\Mm \times \Aa}_\zeta(s) &= \sup_{\f}\big(\EDisct^{\Mm \times \Aa}_\zeta((s,q_0),
  \f)\big) \enspace.
\end{align*}

\begin{theorem}
 For every strategy $\f$, the expected reward for $\overline{\Mm}^\zeta_\f$ is equal to the expected $\zeta$-biased reward for $({\Mm} \times {\Aa})_\f$:
 $\EDisct^{\Mm \times \Aa}_\zeta((s,q), \f)=\ETotal^{\overline{\Mm}^\zeta}((s,q), \f)$. 
\end{theorem}

This is simply because the discounted reward for each transition is equal to the chance of not having reached $t$ before (and thus still seeing this transition) in $\overline{\Mm}^\zeta_\f$.

This improves over \cite{DBLP:journals/corr/abs-1909-07299} because it only uses one discount parameter, $\zeta$, instead of two (called $\gamma$ and $\gamma_B$ in \cite{DBLP:journals/corr/abs-1909-07299}) parameters (that are not independent). 
It is also simpler and more intuitive: discount whenever you have earned a reward.

\end{document}